\documentclass[11pt]{article}
\usepackage{amsmath,amsfonts,amsthm,amssymb,graphicx,bm,color}

\theoremstyle{plain}
\newtheorem{thm}{Theorem}

\newtheorem{lemma}{Lemma}

\newtheorem{assumption}{Assumption}

\begin{document}

\title{The Empirical Likelihood Approach to Quantifying Uncertainty in Sample Average Approximation\footnote{A preliminary conference version of this work has appeared in \cite{lam2015quantifying}. Research of the first author was partially supported by the National Science Foundation under Grants CMMI-1400391/1542020 and CMMI-1436247/1523453. Research of the second author was partially supported by the National Science Foundation under Grant CAREER CMMI-1453934, and Air Force Office of Scientific Research under Grant YIP FA- 9550-14-1-0059.}}
\author{Henry Lam\thanks{Department of Industrial and Operations Engineering, University of Michigan, Ann Arbor, MI. Email: {\tt khlam@umich.edu}}
\and
Enlu Zhou\thanks{H. Milton Stewart School of Industrial and Systems Engineering, Georgia Institute of Technology, Atlanta, GA. Email: {\tt enlu.zhou@isye.gatech.edu}}}
\date{}
\maketitle

\begin{abstract}
We study the empirical likelihood approach to construct confidence intervals for the optimal value and the optimality gap of a given solution, henceforth quantify the statistical uncertainty of sample average approximation, for optimization problems with expected value objectives and constraints where the underlying probability distributions are observed via limited data. This approach relies on two distributionally robust optimization problems posited over the uncertain distribution, with a divergence-based uncertainty set that is suitably calibrated to provide asymptotic statistical guarantees. 
\\


\end{abstract}

\section{Introduction}
We consider a stochastic optimization problem in the form
\begin{equation} \label{obj}
\min_{x\in\Theta}\{h(x):=E[H(x;\xi)]\},
\end{equation}
where $x=(x_1,\ldots,x_p)$ is a continuous decision variable in the deterministic feasible region $\Theta\subseteq\mathbb{R}^p$, and $\xi$ is a random vector on $\mathbb{R}^d$. We are interested in situations where the underlying probability distribution that controls the expectation $E[\cdot]$ is not fully known and can only be accessed via limited data $\xi_1,\ldots,\xi_n$. It is customary in this setting to work on an empirical counterpart of the problem, namely by solving the sample average approximation (SAA) (e.g., \cite{shapiro2014lectures}):
\begin{equation} \label{obj empirical}
\min_{x\in\Theta}\frac{1}{n}\sum_{i=1}^nH(x;\xi_i).
\end{equation}

We further consider problems with expected value constraints, in the form
\begin{equation} \label{obj constrained}
\begin{array}{ll}\min&h(x)=E[H(x;\xi)]\\
\text{subject to}&f_k(x)=E[F_k(x;\xi)]\leq0,\ k=1,\ldots,m\\
&g_k(x)\leq0,\ k=1,\ldots,s
\end{array}
\end{equation}
where $g_k(\cdot)$'s are deterministic functions. Thus \eqref{obj constrained} can include both stochastic and deterministic constraints. Again, under limited data $\xi_1,\ldots,\xi_n$, an SAA version of \eqref{obj constrained} is (e.g., \cite{wang2008sample})
\begin{equation} \label{obj constrained empirical}
\begin{array}{ll}\min&\frac{1}{n}\sum_{i=1}^nH(x;\xi_i)\\
\text{subject to}&\frac{1}{n}\sum_{i=1}^nF_k(x;\xi_i)\leq0,\ k=1,\ldots,m\\
&g_k(x)\leq0,\ k=1,\ldots,s
\end{array}
\end{equation}

Our premise is that beyond the $n$ observations, new samples are not easily accessible because of either a lack of data or limited computational capacity in running further Monte Carlo simulation. The optimal value and solution obtained from \eqref{obj empirical} and \eqref{obj constrained empirical} thus deviate from those under the genuine distribution in \eqref{obj} and \eqref{obj constrained}. Moreover, the error of the solution implies a non-zero optimality gap with the true optimal value, resulting in suboptimal decisions. Estimating these errors is important and has been studied over the years (e.g., \cite{kleywegt2002sample}, \cite{mak1999}, Chapter 5 in \cite{shapiro2014lectures}).



Our main contribution is to bring in a new approach to rigorously quantify the uncertainty in \eqref{obj empirical} and \eqref{obj constrained empirical} through constructing confidence intervals (CIs) for the true optimal value and the optimality gap for a given solution. The machinery underlying our framework uses the so-called empirical likelihood (EL) method in statistics, and culminates at a reformulation of the problem of finding the upper and lower bounds of a CI into solving two optimization problems that closely resemble distributionally robust optimization (DRO). The uncertainty set in the DRO is a divergence-based ball cast over an uncertain probability distribution, where the size of the ball is suitably calibrated so that it provides asymptotic guarantees for the coverage probability of the resulting CI.

We study the theory giving rise to such guarantees. We demonstrate through several numerical examples that our method compares favorably with some existing methods, such as bounds using the central limit theorem (CLT) and the delta method, in terms of finite-sample performance. In the remainder of this paper, Sections \ref{EL section} and \ref{optimality gap section} study the theory of our approach applied to the optimal value and the optimality gap, and Section \ref{sec:numerics} shows some numerical results and comparison with previous methods.

\section{The Empirical Likelihood Method for Constructing Confidence Bounds for Optimal Values}\label{EL section}
This section studies in detail the EL method in constructing CIs for the optimal values. Section \ref{sec:unconstrained} focuses on \eqref{obj} that only has deterministic constraints, and Section \ref{sec:constrained} generalizes to the stochastically constrained case \eqref{obj constrained}.

\subsection{Deterministically Constrained Optimization}\label{sec:unconstrained}
Let us first fix some notations. Given the set of i.i.d. data $\xi_1,\xi_2,\ldots,\xi_n$, we denote a probability vector over $\{\xi_1,\ldots,\xi_n\}$ as $w=(w_1,\ldots,w_n)\in\mathbb R^n$, where $\sum_{i=1}^nw_i=1$ and $w_i\geq0$ for all $i=1,\ldots,n$. We denote $\chi^2_{q,\beta}$ as the $1-\beta$ quantile of a $\chi^2$ distribution with degree of freedom $q$. We use ``$\Rightarrow$" to denote convergence in distribution, and ``a.s." to denote ``almost surely".

Our method utilizes the optimization problems
\begin{equation}
\begin{array}{ll}
\max/\min_w&\min_{x\in\Theta}\sum_{i=1}^nw_iH(x;\xi_i)\\
\text{subject to}&-2\sum_{i=1}^n\log(nw_i)\leq\chi^2_{p+1,\beta}\\
&\sum_{i=1}^nw_i=1\\
&w_i\geq0\text{\ for all\ }i=1,\ldots,n
\end{array}\label{EL opt}
\end{equation}
where ``$\max/\min$" denotes a pair of maximization and minimization. Note that the optimal value of the SAA problem \eqref{obj empirical} lies between those of \eqref{EL opt}.

The quantity $-(1/n)\sum_{i=1}^n\log(nw_i)$ can be interpreted as the Burg-entropy divergence (\cite{pardo2005statistical}, \cite{ben2013robust}) between the probability distributions represented by the weights $w$ and by the uniform weights $(1/n)_{i=1,\ldots,n}$ on the support $\{\xi_1,\ldots,\xi_n\}$. 
Thus, the first constraint in \eqref{EL opt} is a Burg-entropy divergence ball centered at the uniform weights, with radius $\chi^2_{p+1,\beta}/(2n)$. From the viewpoint of DRO (e.g., \cite{delage2010distributionally,ben2013robust,wiesemann2014distributionally}), the optimization problems in \eqref{EL opt} output the worst-case estimates of $\min_{x\in\Theta}\{h(x)=E[H(x;\xi)]\}$ when $E[\cdot]$ is uncertain and its underlying distribution is believed to lie inside the divergence ball. We should point out, however, that this DRO interpretation differs from those in the existing literature (e.g., \cite{bertsimas2016}), as our divergence ball (i.e. the ``uncertainty set" in the terminology of robust optimization) may have low coverage of the true distribution $P$. This can be seen particularly when $P$ is a continuous distribution, in which case the coverage of the divergence ball is zero because of the violation of the absolute continuity requirement needed in properly defining the divergence.


The EL method is a mechanism to endow statistical meaning to \eqref{EL opt}. In particular, it asserts that using the ball size $\chi^2_{p+1,\beta}/(2n)$ in \eqref{EL opt} gives rise to statistically valid $1-\beta$ confidence bounds for the optimal value of \eqref{obj} (despite that the ball may under-cover the true distribution). This method originates as a nonparametric analog of maximum likelihood estimation first proposed by \cite{owen1988empirical}. On the data set $\{\xi_1,\ldots,\xi_n\}$, we first define a ``nonparametric likelihood" $\prod_{i=1}^nw_i$, where $w_i$ is a probability weight applied to each datum. It is straightforward to see that the maximum value of $\prod_{i=1}^nw_i$, among all $w$ in the probability simplex, is $\prod_{i=1}^n(1/n)$. In fact, the same conclusion holds even if one allows putting weights outside the support of the data, which could only make the likelihood $\prod_{i=1}^nw_i$ smaller. In this sense, $\prod_{i=1}^n(1/n)$ can be viewed as a maximum likelihood in the nonparametric space. Correspondingly, we define the nonparametric likelihood ratio between the weights $w$ and the maximum likelihood weights as $\prod_{i=1}^nw_i/\prod_{i=1}^n(1/n)=\prod_{i=1}^n(nw_i)$. 


The key of the EL method is a nonparametric counterpart of the celebrated Wilks' Theorem \cite{wilks1938large} in parametric likelihood inference. The latter states that the ratio between the maximum likelihood and the true likelihood (the parametric likelihood ratio) converges to a $\chi^2$-distribution in a suitable logarithmic scale. To develop this analog, we first incorporate a target parameter of interest, i.e. the quantity whose statistical uncertainty is to be assessed (or to be ``estimated"). Say this parameter is $\theta\in\mathbb R^p$. 
Suppose the true parameter is known to satisfy the set of equations $E[t(\theta;\xi)]=0$ where $E[\cdot]$ is the expectation for the random object $\xi\in\mathbb R^d$, and $t(\theta;\xi),0\in\mathbb R^b$. We define the nonparametric profile likelihood ratio as
\begin{equation}
\mathcal R(\theta)=\max\left\{\prod_{i=1}^nnw_i:\sum_{i=1}^nw_it(\theta;\xi_i)=0,\ \sum_{i=1}^nw_i=1,\ w_i\geq0\text{\ for all\ }i=1,\ldots,n\right\}\label{profile1}
\end{equation}
where profiling refers to the categorization of all weights that respect the set of equations $E[t(\theta;\xi)]=0$.

With the above definitions, the crux is the empirical likelihood theorem (ELT):
\begin{thm}[Theorem 3.4 in \cite{owen2001empirical}]
Let $\xi_1,\ldots,\xi_n\in\mathbb R^d$ be i.i.d. data. Let $\theta_0\in\mathbb R^p$ be a value of the parameter that satisfies $E[t(\theta;\xi)]=0$, where $t(\theta;\xi),0\in\mathbb R^b$. Assume the covariance matrix $Var(t(\theta_0;\xi))$ is finite and has rank $q>0$. Then $-2\log\mathcal R(\theta_0)\Rightarrow\chi^2_q$, where $\mathcal R(\theta)$ is defined in \eqref{profile1}.\label{ELT estimating equations}
\end{thm}

The quantity $-2\log\mathcal R(\theta)$ is defined as $\infty$ if the optimization in \eqref{profile1} is infeasible.

We now explain how \eqref{EL opt} provides confidence bounds for optimization problem \eqref{obj}. We make the following assumptions:
\begin{assumption}
\begin{enumerate}
\item $h(x)$ is differentiable in $x$ with $\nabla_xh(x)=E[\nabla_xH(x;\xi)]$ for all $x\in\Theta$.\label{differentiability}
\item $x^*\in\text{argmin}_{x\in\Theta}h(x)$ if and only if $\nabla_xh(x^*)=0$. Moreover, this relation is \emph{distributionally stable}, meaning that $\tilde x^*\in\text{argmin}_{x\in\Theta}\tilde h(x)$ if and only if $\nabla_x\tilde h(\tilde x^*)=0$ for any $\tilde h(x)=\tilde E[H(\tilde x;\xi)]$ that has the expectation $\tilde E[\cdot]$ generated under an arbitrary distribution $\tilde P$ such that
$$\sup_{x\in\Theta}|\tilde h(x)-h(x)|<\epsilon$$
for small enough $\epsilon>0$.
\label{first order}
\item There exists an $x^*\in\text{argmin}_{x\in\Theta}h(x)$ such that the covariance matrix of the random vector $(\nabla_x H(x^*;\xi),H(x^*;\xi))\in\mathbb R^{p+1}$ is finite and has positive rank.\label{rank}
\item
$\frac{1}{n}\sum_{i=1}^nH(x;\xi_i)\to h(x)$ uniformly over $x\in\Theta$ a.s..\label{uniform convergence}
\item $E\left[\sup_{x\in\Theta}H(x;\xi)^2\right]<\infty$\label{second moment}
\end{enumerate}\label{EL assumptions}
\end{assumption}

Assuming the existence of $\nabla_xH(x;\xi)$ a.s., the interchangeability of derivative and expectation in Assumption \ref{EL assumptions}.\ref{differentiability} can generally be justified by the pathwise Lipschitz continuity condition
$$|H((x_1,\ldots,x_{j-1},u,x_{j+1},\ldots,x_p);\xi)-H((x_1,\ldots,x_{j-1},v,x_{j+1},\ldots,x_p);\xi)|\leq M_j|u-v|\text{\ \ a.s.}$$
for any $u,v$ in a nonrandom neighborhood around the point $x_j$ to be differentiated and $M_j$ measurable with $EM_j<\infty$ (e.g., \cite{asmussen2007stochastic}). Another sufficient condition is that $H(x;\xi)$ is a.s. continuous and piecewise differentiable in $x_j$ and $\sup_{u\in D}|(\partial/\partial x_j)H((x_1,\ldots,x_{j-1},u,x_{j+1},\ldots,x_p);\xi)|$ is integrable where $D$ is a neighborhood around $x_j$ \cite{Glasserman:1988:PCD:318123.318245}. Assumption \ref{EL assumptions}.\ref{first order} states that the first order condition for optimality is both sufficient and necessary. Assumptions \ref{EL assumptions}.\ref{first order} and \ref{EL assumptions}.\ref{uniform convergence} together ensure that this first order condition is unchanged when the true distribution is replaced by a (weighted) empirical version as the sample size gets large. Assumption \ref{EL assumptions}.\ref{second moment} is a technical condition required to bound the error between the empirical distribution and its weighted version within the divergence ball. Assumption \ref{EL assumptions}.\ref{rank} is used to invoke Theorem \ref{ELT estimating equations}. Note that $x^*$ is not necessarily unique.

As our subsequent development will reveal, both the necessity and the sufficiency of the first order condition in Assumption \ref{EL assumptions}.\ref{first order} are required; in particular, we need the necessity of $\nabla_xh(x^*)=0$ for $x^*\in\text{argmin}_{x\in\Theta}h(x)$ and the sufficiency of $\nabla_x\tilde h(\tilde x^*)=0$ for $\tilde x^*\in\text{argmin}_{x\in\Theta}\tilde h(x)$ in order for our argument on statistical guarantee to go through. 
Assumptions \ref{EL assumptions}.\ref{first order}, \ref{EL assumptions}.\ref{uniform convergence} and \ref{EL assumptions}.\ref{second moment} can be replaced by a single condition
\begin{assumption}
$x^*\in\text{argmin}_{x\in\Theta}h(x)$  if and only if $\nabla_xh(x^*)=0$, and $\tilde x^*\in\text{argmin}_{x\in\Theta}\sum_{i=1}^nw_iH(x;\xi_i)$ if and only if $\sum_{i=1}^nw_i\nabla_xH(\tilde x^*;\xi_i)=0$ for any support set $\{\xi_1,\ldots,\xi_n\}\subset\Theta$ and arbitrary probability vector $w$.\label{first order alternate}
\end{assumption}
Assumption \ref{first order alternate} is satisfied by, for instance, $H(\cdot;\xi)$ that is coersive and convex for any $\xi$ and $\Theta=\mathbb R^p$.


We have the following statistical guarantee:
\begin{thm}
Suppose $\xi_1,\ldots,\xi_n\in\mathbb R^d$ are i.i.d. data. Let $z^*$ be the optimal value of \eqref{obj}, and $\overline z$ and $\underline z$ be the maximum and minimum values of \eqref{EL opt} respectively. Then, under Assumption \ref{EL assumptions}, we have
$$\liminf_{n\to\infty}P\left(z^*\in[\underline z,\overline z]\right)\geq1-\beta.$$\label{EL thm}
\end{thm}
\begin{proof}
By Assumption \ref{EL assumptions}.\ref{rank}, there exists an $x^*\in\text{argmin}_{x\in\Theta}h(x)$ such that the covariance matrix of the random vector $(\nabla_x H(x^*;\xi),H(x^*;\xi))$ has a positive rank, call it $r$. Also, by Assumptions \ref{EL assumptions}.\ref{differentiability} and \ref{EL assumptions}.\ref{first order}, $x^*$ satisfies $E[\nabla_x H(x^*;\xi)]=0$. 

We define the nonparametric profile likelihood ratio as
\begin{equation}
\mathcal R(x,z)=\max\left\{\prod_{i=1}^nnw_i:\sum_{i=1}^nw_i\nabla_x H(x;\xi_i)=0,\ \sum_{i=1}^nw_iH(x;\xi_i)=z,\ \sum_{i=1}^nw_i=1,\ w_i\geq0\text{\ for all\ }i=1,\ldots,n\right\}.\label{profile}
\end{equation}
Let $z^*=\min_{x\in\Theta}h(x)=h(x^*)$. From Theorem \ref{ELT estimating equations}, the nonparametric profile likelihood ratio \eqref{profile} satisfies $-2\log\mathcal R(x^*,z^*)\Rightarrow\chi^2_r$ as $n\to\infty$. 
This implies $P(-2\log\mathcal R(x^*,z^*)\leq\chi^2_{r,\beta})\to1-\beta$.

The rest of the proof focuses on the event $-2\log\mathcal R(x^*,z^*)\leq\chi^2_{r,\beta}$. Write
\begin{eqnarray}
&&-2\log\mathcal R(x^*,z^*)\notag\\
&=&\min\Bigg\{-2\sum_{i=1}^n\log(nw_i):\sum_{i=1}^nw_i\nabla_x H(x^*;\xi_i)=0,\ \sum_{i=1}^nw_iH(x^*;\xi_i)=z^*,\ \sum_{i=1}^nw_i=1,\ w_i\geq0{}\notag\\
&&{}\text{\ for all\ }i=1,\ldots,n\Bigg\}\label{profile updated}
\end{eqnarray}
We argue that $-2\log\mathcal R(x^*,z^*)\leq\chi^2_{r,\beta}$ implies the existence of a probability vector $w$ such that \begin{equation}
\sum_{i=1}^nw_i\nabla_x H(x^*;\xi_i)=0,\ \sum_{i=1}^nw_iH(x^*;\xi_i)=z^*,\ -2\sum_{i=1}^n\log(nw_i)\leq\chi^2_{r,\beta}\label{properties}
\end{equation}
Notice that $-2\sum_{i=1}^n\log(nw_i)=\infty$ if $w_i=0$ for any $i$. Hence it suffices to replace, in \eqref{profile updated}, $w_i\geq0$ with $w_i\geq\epsilon$ for all $i$, for some small enough $\epsilon>0$. In this modified, compact, feasible set, $-2\sum_{i=1}^n\log(nw_i)$ is bounded and hence must possess an optimal solution $w$, which is a probability vector that satisfies \eqref{properties}.

This further implies that $z^*$ is bounded from above and below by the optimization problems
\begin{equation}
\begin{array}{ll}
\max/\min_{w}&\sum_{i=1}^nw_iH(x^*;\xi_i)\\
\text{subject to}&\sum_{i=1}^nw_i\nabla_x H(x^*;\xi_i)=0\\
&-2\sum_{i=1}^n\log(nw_i)\leq\chi^2_{r,\beta}\\
&\sum_{i=1}^nw_i=1\\
&w_i\geq0\text{\ for all\ }i=1,\ldots,n
\end{array}\label{opt interim}
\end{equation}
We argue that as $n\to\infty$, \eqref{opt interim} is equivalent to
\begin{equation}
\begin{array}{ll}
\max/\min_{w}&\sum_{i=1}^nw_iH(x^*;\xi_i)\\
\text{subject to}&w\in\{(w_1,\ldots,w_n):x^*\in\text{argmin}_{x\in\Theta}\sum_{i=1}^nw_iH(x;\xi_i)\}\\
&-2\sum_{i=1}^n\log(nw_i)\leq\chi^2_{r,\beta}\\
&\sum_{i=1}^nw_i=1\\
&w_i\geq0\text{\ for all\ }i=1,\ldots,n
\end{array}\label{opt interim1}
\end{equation}
eventually (i.e. with probability 1). Note that the first constraint in \eqref{opt interim1} states that the probability vector $w$ must be chosen such that $x^*$, a minimizer of $h(x)$ picked at the beginning of this proof, also minimizes $\sum_{i=1}^nw_iH(x;\xi_i)$.

To develop the argument for the asymptotic equivalence, let us denote $P^w$ as the distribution represented by the probability weights $w$ on the support $\{\xi_1,\ldots,\xi_n\}$. Denote $E^w[\cdot]$ as the associated expectation and $h^w(x)=E^w[H(x;\xi)]$. We will show that
\begin{equation}
\sup_{x\in\Theta,w\in\mathcal W_r}|h^w(x)-h(x)|\to0\text{\ \ a.s.}\label{opt interim2}
\end{equation}
where
\begin{equation}
\mathcal W_r=\left\{(w_1,\ldots,w_n)\in\mathbb R^n:-2\sum_{i=1}^n\log(nw_i)\leq\chi^2_{r,\beta},\ \sum_{i=1}^nw_i=1,\ w_i\geq0\text{\ for all\ }i=1,\ldots,n\right\}\label{constraint set}
\end{equation}
Assumption \ref{EL assumptions}.\ref{first order} then implies that with probability 1, for sufficiently large $n$, $\sum_{i=1}^nw_i\nabla_x H(x^*;\xi_i)=0$ if and only if $x^*\in\text{argmin}_{x\in\Theta}\sum_{i=1}^nw_iH(x;\xi_i)$ for any $w\in\mathcal W_r$, leading to the equivalence.

We now show \eqref{opt interim2}. Consider
\begin{eqnarray}
\sup_{x\in\Theta,w\in\mathcal W_r}|h^w(x)-h(x)|&=&\sup_{x\in\Theta,w\in\mathcal W_r}\left|\int H(x;\xi)d(P^w-P)(\xi)\right|\notag\\
&\leq&\sup_{x\in\Theta,w\in\mathcal W_r}\left|\int H(x;\xi)d(P^w-\hat P)(\xi)\right|+\sup_{x\in\Theta,w\in\mathcal W_r}\left|\int H(x;\xi)d(\hat P-P)(\xi)\right|{}\notag\\
&&{}\text{\ \ where $\hat P$ denotes the empirical distribution generated from $\{\xi_1,\ldots,\xi_n\}$}\notag\\
&\leq&\sup_{x\in\Lambda,1\leq i\leq n,w\in\mathcal W_r}|H(x;\xi_i)|d_{TV}(P^w,\hat P)+\sup_{x\in\Theta}\left|\int H(x;\xi)d(\hat P-P)(\xi)\right|{}\label{interim new3}\\
&&{}\text{\ \ where $d_{TV}$ denotes the total variation distance}\notag
\end{eqnarray}
Now by Lemma 11.5 in \cite{owen2001empirical} (restated in the Appendix) and Assumption \ref{EL assumptions}.\ref{second moment}, we have
\begin{equation}
\sup_{x\in\Theta,1\leq i\leq n}|H(x;\xi_i)|=\max_{1\leq i\leq n}\sup_{x\in\Theta}|H(x;\xi_i)|=o(n^{1/2})\text{\ \ a.s.}\label{interim new1}
\end{equation}
On the other hand, by Pinsker's inequality, for any $w\in\mathcal W_r$,
\begin{equation}
d_{TV}(P^w,\hat P)\leq\sqrt{\frac{d_{KL}(P^w,\hat P)}{2}}=\sqrt{\frac{-\sum_{i=1}^n\log(nw_i)}{2n}}\leq\sqrt{\frac{\chi^2_r}{4n}}\label{interim new2}
\end{equation}
where $d_{KL}$ denotes the Kullback-Leibler divergence. Combining \eqref{interim new1} and \eqref{interim new2}, the first term in \eqref{interim new3} goes to 0 a.s.. The second term in \eqref{interim new3} converges to 0 a.s. by Assumption \ref{EL assumptions}.\ref{uniform convergence}. Hence $\sup_{x\in\Theta}|h^w(x)-h(x)|\to0$ a.s.. Therefore \eqref{opt interim} is equivalent to \eqref{opt interim1} eventually as $n\to\infty$.

Consider \eqref{opt interim1}. With the first constraint, the objective function must be equal to $\min_{x\in\Theta}\sum_{i=1}^nw_iH(x;\xi_i)$. Thus \eqref{opt interim1} is equivalent to
\begin{equation}
\begin{array}{ll}
\max/\min_w&\min_{x\in\Theta}\sum_{i=1}^nw_iH(x;\xi_i)\\
\text{subject to}&w\in\{(w_1,\ldots,w_n):x^*\in\text{argmin}_{x\in\Theta}\sum_{i=1}^nw_iH(x;\xi_i)\}\\
&-2\sum_{i=1}^n\log(nw_i)\leq\chi^2_{r,\beta}\\
&\sum_{i=1}^nw_i=1\\
&w_i\geq0\text{\ for all\ }i=1,\ldots,n
\end{array}\label{opt}
\end{equation}
Let $\overline v$ and $\underline v$ be the maximum and minimum values of \eqref{opt}. Note that $r\leq p+1$ since $r$ is the rank of a $\mathbb R^{(p+1)\times(p+1)}$ matrix. This implies $\chi^2_{r,\beta}\leq\chi^2_{p+1,\beta}$. Together with a relaxation by removing the first constraint in \eqref{opt}, we have $\underline v\geq\underline z$ and $\overline v\leq\overline z$ where $\overline z$ and $\underline z$ are the maximum and minimum values of \eqref{EL opt}. From this we conclude that
$$\liminf_{n\to\infty}P(\underline z\leq z^*\leq\overline z)\geq\liminf_{n\to\infty}P(\underline v\leq z^*\leq\overline v)\geq\liminf_{n\to\infty}P(-2\log\mathcal R(x^*,z^*)\leq\chi^2_{r,\beta})=1-\beta$$

%
\end{proof}

Note that we have obtained bounds by relaxing the constraints in \eqref{opt}, and the degree of freedom in the $\chi^2$-distribution may not be optimally chosen. Nevertheless, our numerical examples show that, at least for small $p$, the EL method provides reasonably tight CIs. There exists techniques (e.g., bootstrap calibration or Bartlett correction; \cite{owen1988empirical,diciccio1991empirical}) that can improve the coverage of the EL method in estimation problems. Investigation of these techniques in the optimization context is delegated to future work.

Note that the equivalence of \eqref{opt interim} and \eqref{opt interim1} holds for all $n$ if Assumption \ref{first order alternate} replaces Assumptions \ref{EL assumptions}.\ref{first order}, \ref{EL assumptions}.\ref{uniform convergence} and \ref{EL assumptions}.\ref{second moment}. 

\subsection{Stochastically Constrained Optimization}\label{sec:constrained}
We generalize the EL method to the stochastically constrained problem \eqref{obj constrained}. In this setting, we construct CI via the following optimization problems
\begin{equation}
\begin{array}{ll}
\max/\min_w&\left\{\begin{array}{ll}\min_x&\sum_{i=1}^nw_iH(x;\xi_i)\\
\text{subject to}&\sum_{i=1}^nw_iF_k(x;\xi_i)\leq0,\ k=1,\ldots,m\\
&g_k(x)\leq0,\ k=1,\ldots,s
\end{array}\right\}\\
\text{subject to}&-2\sum_{i=1}^n\log(nw_i)\leq\chi^2_{p+m+1,\beta}\\
&\sum_{i=1}^nw_i=1\\
&w_i\geq0\text{\ for all\ }i=1,\ldots,n
\end{array}\label{opt constrained}
\end{equation}
While resembling \eqref{EL opt}, we note that the degree of freedom in the $\chi^2$-distribution is now $p+m+1$, which includes the number of stochastic constraints compared to \eqref{EL opt}.

For convenience, we denote
$$\Lambda=\left\{x\in\mathbb R^p:g_k(x)\leq0,\ k=1,\ldots,s\right\}$$
as the set of $x$ satisfying the deterministic constraints in \eqref{obj constrained}.

We make the following assumptions in parallel to Assumption \ref{EL assumptions}:
\begin{assumption}
We assume:
\begin{enumerate}
\item $h(x)=E[H(x;\xi)]$, $f_k(x)=E[F_k(x;\xi)],\ k=1,\ldots,m$ and $g_k(x),\ k=1,\ldots,s$ are all differentiable in $x\in\Lambda$, and
$$\nabla_xh(x)=E[\nabla_xH(x;\xi)],\ \nabla_xf_k(x)=E[\nabla_xF_k(x;\xi)]$$
\item Let $S^*$ be the set of all optimal solutions for \eqref{obj constrained}. $x^*\in S^*$ if and only if $x^*$ satisfies the KKT condition, where the active set of the KKT condition (i.e. equalities) is unique among all $x^*\in S^*$ and is sufficient for determining $S^*$. This relation is \emph{distributionally stable}, meaning that $\tilde x^*\in\tilde S^*$, where $\tilde S^*$ is the set of optimal solutions for
\begin{equation} \label{obj constrained perturbed}
\begin{array}{ll}\min&\tilde h(x)\\
\text{subject to}&\tilde f_k(x)\leq0,\ k=1,\ldots,m\\
&\tilde g_k(x)\leq0,\ k=1,\ldots,s
\end{array}
\end{equation}
if and only if $\tilde x^*$ satisfies the corresponding KKT condition, where $\tilde h(x)=\tilde E[H(x;\xi)]$ and $\tilde f_k(x)=\tilde E[F_k(x;\xi)]$, with $\tilde E$ denoting the expectation under an arbitrary distribution $\tilde P$ such that
$$\sup_{x\in\Lambda}|\tilde h(x)-h(x)|<\epsilon$$
$$\sup_{x\in\Lambda}|\tilde f_k(x)-f_k(x)|<\epsilon\text{\ \ for all\ }k=1,\ldots,m$$
for small enough $\epsilon>0$. Moreover, for any such $\epsilon>0$, the active set of the KKT condition at any $\tilde x^*\in\tilde S^*$ for \eqref{obj constrained perturbed} is the same as that at any $x^*\in S^*$ for \eqref{obj constrained} and is sufficient for determining $\tilde S^*$.
%
\item
  There exists an optimal solution $x^*$ for \eqref{obj constrained}, with associated Lagrange multipliers for the stochastic constraints in \eqref{obj constrained} given by $\lambda^*=(\lambda_1^*,\ldots,\lambda_m^*)$, such that the covariance matrix of the variables $H(x^*;\xi)$, $\frac{\partial}{\partial x_j}H(x^*;\xi)+\sum_{k=1}^m\lambda_k^*\frac{\partial}{\partial x_j}F_k(x^*;\xi)$, and $F_k(x^*;\xi)$, for all indices $j$ and $k$ corresponding to the active set of the KKT condition, is finite and has positive rank.
\item
$$\frac{1}{n}\sum_{i=1}^nH(x;\xi_i)\to h(x)\text{\ \ and\ \ }\frac{1}{n}\sum_{i=1}^nF_k(x;\xi_i)\to f_k(x),\ \ k=1,\ldots,m$$
uniformly over $x\in\Lambda$ a.s..
\item $E\left[\sup_{x\in\Theta}H(x;\xi)^2\right]<\infty$ and $E\left[\sup_{x\in\Theta}F_k(x;\xi)^2\right]<\infty$ for $k=1,\ldots,m$.
\end{enumerate}\label{assumptions constrained}
\end{assumption}

Denote $\nu^*=(\nu_1^*,\ldots,\nu_s^*)$ as the Lagrange multiplier for the deterministic constraints in \eqref{obj constrained}. In Assumptions \ref{assumptions constrained}.2 and \ref{assumptions constrained}.3 above, the active set of the KKT condition satisfied by $(x^*,\lambda^*,\nu^*)$ is in the form
$$\frac{\partial}{\partial x_j}h(x^*)+\sum_{k=1}^m\lambda_k^*\frac{\partial}{\partial x_j}f_k(x^*)+\sum_{k=1}^s\nu_k^*\frac{\partial}{\partial x_j}g_k(x^*)=0,\ j\in\mathcal A_1^*\equiv\{1,\ldots,p\}$$
  $$f_k(x^*)=0,\ k\in\mathcal A_2^*\subset\{1,\ldots,m\}$$
  $$g_k(x^*)=0,\ k\in\mathcal A_3^*\subset\{1,\ldots,s\}$$
  $$\lambda_k^*=0\ k\in\{1,\ldots,m\}\setminus\mathcal A_2^*$$
  $$\nu_k^*=0,\ \ k\in\{1,\ldots,s\}\setminus\mathcal A_3^*$$
where $\mathcal A_1^*$, $\mathcal A_2^*$ and $\mathcal A_3^*$ denote the sets of indices that correspond to the equalities in the condition, which are unique among any optimal solutions of \eqref{obj constrained} by Assumption \ref{assumptions constrained}.2. The $j$ and $k$ described in Assumption \ref{assumptions constrained}.3 refer to the indices in $\mathcal A_1^*$ and $\mathcal A_2^*$. Assumption \ref{assumptions constrained}.2 further enforces the sets $\mathcal A_1^*$, $\mathcal A_2^*$ and $\mathcal A_3^*$ to remain as the active sets under a perturbation to $\tilde P$ described therein, and the equalities indexed via these sets are enough to determine $S^*$ and $\tilde S^*$. Assumptions \ref{assumptions constrained}.2 and \ref{assumptions constrained}.3 generalize Assumptions \ref{EL assumptions}.\ref{first order} and \ref{EL assumptions}.\ref{rank} from the first order condition to the KKT condition. Similar to Section \ref{sec:unconstrained}, we require the necessity of the KKT and the active set conditions regarding \eqref{obj constrained} and the sufficiency regarding \eqref{obj constrained perturbed} for our development to go through. Constraint qualification for the validity of the KKT condition is implicitly assumed in Assumption \ref{assumptions constrained}.2.


We have the following result:
\begin{thm}
Suppose $\xi_1,\ldots,\xi_n$ are i.i.d. data. Under Assumption \ref{assumptions constrained}, we have
$$\liminf_{n\to\infty}P(z^*\in[\underline z,\overline z])\geq1-\beta$$
where $z^*$ is the optimal value of \eqref{obj constrained}, and $\underline z$ and $\overline z$ are the minimum and maximum values of \eqref{opt constrained}.
\end{thm}
%
%

\begin{proof}
Consider the nonparametric profile likelihood ratio
\begin{equation}
\mathcal R(x,\lambda,\nu,z)=\max\left\{\prod_{i=1}^nnw_i:\begin{array}{l}\sum_{i=1}^nw_iH(x;\xi_i)=z\\
\sum_{i=1}^nw_i\left(\frac{\partial}{\partial x_j}H(x;\xi_i)+\sum_{k=1}^m\lambda_k\frac{\partial}{\partial x_j}F_k(x;\xi_i)\right)+\sum_{k=1}^s\nu_k\frac{\partial}{\partial x_j}g_k(x)=0,\ j\in\mathcal A_1^*\\
\sum_{i=1}^nw_iF_k(x;\xi_i)=0,\ k\in\mathcal A_2^*\\
g_k(x)=0,\ k\in\mathcal A_3^*\\
\lambda_k=0\ k\in\{1,\ldots,m\}\setminus\mathcal A_2^*\\
\nu_k=0,\ \ k\in\{1,\ldots,s\}\setminus\mathcal A_3^*\\
\sum_{i=1}^nw_i=1\\
w_i\geq0\text{\ for all\ }i=1,\ldots,n
\end{array}\right\}\label{profile constrained}
\end{equation}
Let $x^*$ be an optimal solution for \eqref{obj constrained} satisfying Assumption \ref{assumptions constrained}.3, and $\lambda^*=(\lambda_1^*,\ldots,\lambda_m^*),\nu^*=(\nu_1^*,\ldots,\nu_s^*)$ be its associated Lagrange multipliers. 
  By Assumption \ref{assumptions constrained}.3, the covariance of the random vector concatenated by
  \begin{eqnarray*}
  &&H(x^*;\xi)\\
  &&\frac{\partial}{\partial x_j}H(x^*;\xi)+\sum_{k=1}^m\lambda_k^*\frac{\partial}{\partial x_j}F_k(x^*;\xi)+\sum_{k=1}^s\nu_k^*\frac{\partial}{\partial x_j}g_k(x^*)\text{\ \ for\ }j\in\mathcal A_1^*\\
  &&F_k(x^*;\xi)\text{\ \ for\ }k\in\mathcal A_2^*\\
  \end{eqnarray*}
  has rank $r$ for some $r>0$. Let $z^*$ be the optimal value of \eqref{obj constrained} equal to $h(x^*)$. Since the other active KKT conditions are deterministic, Theorem \ref{ELT estimating equations} implies that $-2\log\mathcal R(x^*,\lambda^*,\nu^*,z^*)\Rightarrow\chi^2_r$, 
which further implies $P(-2\log\mathcal R(x^*,\lambda^*,\nu^*,z^*)\leq\chi^2_{r,\beta})\to1-\beta$.


Similar to the proof of Theorem \ref{EL thm}, $-2\log\mathcal R(x^*,\lambda^*,\nu^*,z^*)\leq\chi^2_{r,\beta}$ implies the existence of a $w$ that satisfies $-2\sum_{i=1}^n\log(nw_i)\leq\chi^2_{r,\beta}$ and all constraints in \eqref{profile constrained} evaluated at $x^*,\lambda^*,\nu^*,z^*$. This in turn implies that $z^*$ is bounded by
\begin{equation}
\begin{array}{ll}
\max/\min_w&\sum_{i=1}^nw_iH(x^*;\xi_i)\\
\text{subject to}&\sum_{i=1}^nw_i\left(\frac{\partial}{\partial x_j}H(x^*;\xi_i)+\sum_{k=1}^m\lambda_k^*\frac{\partial}{\partial x_j}F_k(x^*;\xi_i)\right)+\sum_{k=1}^s\nu_k^*\frac{\partial}{\partial x_j}g_k(x^*)=0,\ j\in\mathcal A_1^*\\
&\sum_{i=1}^nw_iF_k(x^*;\xi_i)=0,\ k\in\mathcal A_2^*\\
&g_k(x^*)=0,\ k\in\mathcal A_3^*\\
&\lambda_k^*=0\ k\in\{1,\ldots,m\}\setminus\mathcal A_2^*\\
&\nu_k^*=0,\ \ k\in\{1,\ldots,s\}\setminus\mathcal A_3^*\\
&-2\sum_{i=1}^n\log(nw_i)\leq\chi^2_{r,\beta}\\
&\sum_{i=1}^nw_i=1\\
&w_i\geq0\text{\ for all\ }i=1,\ldots,n
\end{array}\label{interim opt constrained}
\end{equation}
Using the same argument as in the proof of Theorem \ref{EL thm}, we obtain from Assumptions \ref{assumptions constrained}.4 and \ref{assumptions constrained}.5 that
$$\sup_{x\in\Lambda,w\in\mathcal W_r}|h^w(x)-h(x)|\to0\text{\ \ a.s.}$$
$$\sup_{x\in\Lambda,w\in\mathcal W_r}|f_k^w(x)-f_k(x)|\to0\text{\ \ a.s. for all\ }k=1,\ldots,m$$
where $\mathcal W_r$ is defined in \eqref{constraint set}, and $h^w(x)=E^w[H(x;\xi)]$, $f_k^w(x)=E^w[F_k(x;\xi)]$ with $E^w$ denoting the expectation with respect to $P^w$, the probability distribution represented by the weights $w$ on the support $\{\xi_1,\ldots,\xi_n\}$. 
Thus, by Assumption \ref{assumptions constrained}.2, the set of active KKT conditions for an optimal solution of the weighted sample problem
$$\begin{array}{ll}\min_x&\sum_{i=1}^nw_iH(x;\xi_i)\\
\text{subject to}&\sum_{i=1}^nw_iF_k(x;\xi_i)\leq0,\ k=1,\ldots,m\\
&g_k(x)\leq0,\ k=1,\ldots,s
\end{array}$$
for any $w\in\mathcal W_r$ is identical to that for $x^*$ for \eqref{obj constrained} eventually as $n\to\infty$, and Assumption \ref{assumptions constrained}.2 further implies that \eqref{interim opt constrained} is equivalent to
\begin{equation}
\begin{array}{ll}
\max/\min_w&\sum_{i=1}^nw_iH(x^*;\xi_i)\\
\text{subject to}&w\in\left\{(w_1,\ldots,w_n):x^*\in\left\{\begin{array}{ll}\text{argmin}_x&\sum_{i=1}^nw_iH(x;\xi_i)\\
\text{subject to}&\sum_{i=1}^nw_iF_k(x;\xi_i)\leq0,\ k=1,\ldots,m\\
&g_k(x)\leq0,\ k=1,\ldots,s
\end{array}\right\}\right\}\\
&-2\sum_{i=1}^n\log(nw_i)\leq\chi^2_{r,\beta}\\
&\sum_{i=1}^nw_i=1\\
&w_i\geq0\text{\ for all\ }i=1,\ldots,n
\end{array}\label{interim opt new1 constrained}
\end{equation}
eventually as $n\to\infty$. With the first constraint, the objective function in \eqref{interim opt new1 constrained} must be equal to $\min_x\{\sum_{i=1}^nw_iH(x;\xi_i):\sum_{i=1}^nw_iF_k(x;\xi_i)\leq0,k=1,\ldots,m,\ g_k(x)\leq0,k=1,\ldots,s\}$. Note that
\begin{equation}
r\leq1+|\mathcal A_1^*|+|\mathcal A_2^*|\leq1+p+m\label{df}
\end{equation}
where $|\cdot|$ denotes cardinality. This implies that $\chi^2_{r,\beta}\leq\chi^2_{p+m+1,\beta}$. Thus, together with a relaxation of the first constraint in \eqref{interim opt new1 constrained}, the same argument as in the proof of Theorem \ref{EL thm} stipulates that the maximum and minimum values of \eqref{interim opt new1 constrained} are bounded from above and below respectively by those of \eqref{opt constrained} and concludes the theorem.

\end{proof}

Note that, much like the proof of Theorem \ref{EL thm}, we have relaxed constraints and placed a conservative bound on the degree of freedom of the $\chi^2$-distribution in \eqref{df}, which could potentially be improved with more refined analysis.

\section{The Empirical Likelihood Method for Constructing Confidence Bounds for Optimality Gaps}\label{optimality gap section}
We study the construction of CI for the optimality gap of a given solution using the EL method. Suppose $\hat x$ is obtained from some procedure independently of the data $\xi_1,\ldots,\xi_n$. The optimality gap of $\hat x$ is given by $\mathcal G(\hat x)=h(\hat x)-z^*$ where $z^*$ is the optimal value of either \eqref{obj} or \eqref{obj constrained}. We will show how we can apply the results in Section \ref{EL section} to find the CI for $\mathcal G(\hat x)$.

Consider the optimization problems
\begin{equation}
\begin{array}{ll}
\max/\min_w&\max_{x\in\Theta}\sum_{i=1}^nw_i[H(\hat x;\xi_i)-H(x;\xi_i)]\\
\text{subject to}&-2\sum_{i=1}^n\log(nw_i)\leq\chi^2_{p+1,\beta}\\
&\sum_{i=1}^nw_i=1\\
&w_i\geq0\text{\ for all\ }i=1,\ldots,n
\end{array}\label{EL optimality gap}
\end{equation}

We have the following guarantee in using \eqref{EL optimality gap} to construct the CI for $\mathcal G(\hat x)$ for \eqref{obj}:
\begin{thm}
Suppose $\xi_1,\ldots,\xi_n\in\mathbb R^d$ are i.i.d. data independent of a given solution $\hat x$. Let $\mathcal G(\hat x)$ be the optimality gap of $\hat x$ for \eqref{obj}, and $\overline z$ and $\underline z$ be the maximum and minimum values of the programs in \eqref{EL optimality gap} respectively. Suppose Assumption \ref{EL assumptions} holds except that in Condition \ref{rank}, we consider the covariance matrix of $(\nabla_x H(x^*;\xi),H(x^*;\xi)-H(\hat x;\xi))$ instead. We have
\begin{equation}
\liminf_{n\to\infty}P\left(\mathcal G(\hat x)\in[\underline z,\overline z]\right)\geq1-\beta.\label{guarantee optimality gap}
\end{equation}
\label{EL thm optimality gap}
\end{thm}

\begin{proof}
Let $\bar H(x;\xi)=H(x;\xi)-H(\hat x;\xi)$, and $\bar h(x)=E[\bar H(x;\xi)]=h(x)-h(\hat x)$. We verify that Assumption \ref{EL assumptions}, with the change that the covariance matrix of $(\nabla_x H(x^*;\xi),H(x^*;\xi)-H(\hat x;\xi))$ is considered instead in Condition \ref{rank}, implies that $\bar h$ and $\bar H$ satisfies Assumption \ref{EL assumptions} too with $h$ and $H$ replaced by $\bar h$ and $\bar H$ and $\epsilon$ replaced by $2\epsilon$.

\noindent\underline{Condition 1:} We have $\nabla_x\bar h(x)=\nabla_x(h(x)-h(\hat x))=\nabla_xh(x)=E[\nabla_xH(x;\xi)]=E[\nabla_x(H(x;\xi)-H(\hat x;\xi))]=E[\nabla_x\bar H(x;\xi)]$.

\noindent\underline{Condition 2:} We have $x^*\in\text{argmin}_{x\in\Theta}\bar h(x)\Leftrightarrow x^*\in\text{argmin}_{x\in\Theta}h(x)\Leftrightarrow\nabla_xh(x^*)=0\Leftrightarrow\nabla_x\bar h(x^*)=0$. Similarly, $x^*\in\text{argmin}_{x\in\Theta}\bar{\tilde h}(x)\Leftrightarrow x^*\in\text{argmin}_{x\in\Theta}\tilde h(x)\Leftrightarrow\nabla_x\tilde h(x^*)=0\Leftrightarrow\nabla_x\bar{\tilde h}(x^*)=0$ for any $\bar{\tilde h}(x)=\tilde h(x)-\tilde h(\hat x)$ that satisfies
$$\sup_{x\in\Theta}|\bar{\tilde h}(x)-\bar h(x)|\leq\sup_{x\in\Theta}|\tilde h(x)-h(x)|+|\tilde h(\hat x)-h(\hat x)|<2\epsilon$$

\noindent\underline{Condition 3:} By our modification of this condition we have the covariance of $(\nabla_x\bar H(x^*;\xi),\bar H(x^*;\xi))=(\nabla_xH(x^*;\xi),H(x^*;\xi)-H(\hat x;\xi))$ finite and having a positive rank.

\noindent\underline{Condition 4:} It is straightforward to show that $\frac{1}{n}\sum_{i=1}^n\bar H(x;\xi_i)=\frac{1}{n}\sum_{i=1}^nH(x;\xi_i)-\frac{1}{n}\sum_{i=1}^nH(\hat x;\xi_i)\to0$ a.s. uniformly over $x\in\Theta$.

\noindent\underline{Condition 5:} We have $E[\sup_{x\in\Theta}\bar H(x;\xi)^2]=E[\sup_{x\in\Theta}(H(x;\xi)-H(\hat x;\xi))^2]\leq 4(E[\sup_{x\in\Theta}H(x;\xi)^2]+E[H(\hat x;\xi)^2])<\infty$.

We have therefore verified our claim. Using Theorem \ref{EL thm}, we get that
$$\liminf_{n\to\infty}P(\underline v\leq\bar h(x^*)\leq\overline v)\geq1-\beta$$
where $\overline v$ and $\underline v$ are the maximum and minimum values of
\begin{equation}
\begin{array}{ll}
\max/\min_w&\min_{x\in\Theta}\sum_{i=1}^nw_i\bar H(x;\xi_i)\\
\text{subject to}&-2\sum_{i=1}^n\log(nw_i)\leq\chi^2_{p+1,\beta}\\
&\sum_{i=1}^nw_i=1\\
&w_i\geq0\text{\ for all\ }i=1,\ldots,n
\end{array}\label{EL optimality gap1}
\end{equation}
Noting that $\mathcal G(x^*)=-\bar h(x^*)$, we get \eqref{guarantee optimality gap} immediately.
\end{proof}

For the optimality gap of the stochastically constrained problem \eqref{obj constrained}, we use the following optimization problems
\begin{equation}
\begin{array}{ll}
\max/\min_w&\left\{\begin{array}{ll}\max_x&\sum_{i=1}^nw_i[H(\hat x;\xi_i)-H(x;\xi_i)]\\
\text{subject to}&\sum_{i=1}^nw_iF_k(x;\xi_i)\leq0,\ k=1,\ldots,m\\
&g_k(x)\leq0,\ k=1,\ldots,s
\end{array}\right\}\\
\text{subject to}&-2\sum_{i=1}^n\log(nw_i)\leq\chi^2_{p+m+1,\beta}\\
&\sum_{i=1}^nw_i=1\\
&w_i\geq0\text{\ for all\ }i=1,\ldots,n
\end{array}\label{optimality gap constrained}
\end{equation}

We have the following guarantee in parallel to the deterministically constrained case:
\begin{thm}
Suppose $\xi_1,\ldots,\xi_n\in\mathbb R^d$ are i.i.d. data independent of a given solution $\hat x$. Let $\mathcal G(\hat x)$ be the optimality gap of $\hat x$ for \eqref{obj constrained}, and $\overline z$ and $\underline z$ be the maximum and minimum values of the programs \eqref{optimality gap constrained} respectively. Suppose Assumption \ref{assumptions constrained} holds except that in Condition 3, we consider the covariance matrix of $H(x^*;\xi)-H(\hat x;\xi)$, $\frac{\partial}{\partial x_j}H(x^*;\xi)+\sum_{k=1}^m\lambda_k^*\frac{\partial}{\partial x_j}F_k(x^*;\xi)$, and $F_k(x^*;\xi)$, for all indices $j$ and $k$ corresponding to the active set of the KKT condition for \eqref{obj constrained}. We have
$$\liminf_{n\to\infty}P\left(\mathcal G(\hat x)\in[\underline z,\overline z]\right)\geq1-\beta.$$\label{EL thm optimality gap constrained}
\end{thm}

\begin{proof}
The proof follows verbatim from that of Theorem \ref{EL thm optimality gap}, and noting that the operations involving $f_k$, $F_k$ and $g_k$ are unaffected by the substitution of $h$ and $H$ with $\bar h$ and $\bar H$.
\end{proof}

\section{Numerical Examples}\label{sec:numerics}
We test the presented method numerically on three examples. For proof of concept, the first example is a simple unconstrained quadratic optimization problem. Then we apply the proposed method to
two more examples, including the problem of estimating Conditional-Value-at-Risk (CVaR) and a stochastically constrained portfolio optimization problem. The latter examples strictly speaking do not satisfy our assumptions, since the first order condition or the KKT condition does not hold for their sample counterparts. However, given that the EL method does not rely on these conditions procedurally, we can still test its performance on these examples.

We compare EL with the CIs obtained from the CLT and the delta method (\cite{shapiro2014lectures}, Theorem 5.7). For deterministically constrained problems in the form \eqref{obj}, the ($1-\beta$) CI on the optimal value is given by
\begin{equation}\label{CLT CI}
\left[\hat z_n^*\pm z_{1-\beta/2}\frac{\hat\sigma(\hat x_n^*)}{\sqrt n}\right]
\end{equation}
where $z_{1-\beta/2}$ is the critical value of the standard normal distribution with confidence $1-\beta$, $\hat x^*_n$ is the empirical optimal solution obtained from \eqref{obj empirical}, $\hat z^*_n = (1/n)\sum_{i=1}^nH(\hat x^*;\xi_i)$ is the empirical optimal value, and $\hat\sigma(\hat x^*_n) = \sqrt{(1/(n-1))\sum_{i=1}^n(H(\hat x^*_n; \xi_i)-\hat z^*_n)^2}$ is the empirical standard deviation of $H(\hat x^*_n;\xi)$. Since $\hat z^*_n$ is a low biased estimator of $z^*$, the CI (\ref{CLT CI}) suffers from the under coverage issue. So we also compare with a 2-sample CLT (CLT2) method, as suggested by \cite{mak1999}, which uses first half of the data to compute the empirical optimal value and solution, and then uses the remaining half of the data to estimate the objective value fixed at the solution to generate an upper bound. The 2-sample CLT CI is given by
\begin{equation} \label{CLT2 CI}
\left[\hat z^*_{n/2} - z_{1-\beta/2}\frac{\hat\sigma(\hat x^*_{n/2})}{\sqrt {n/2}}, \bar{z}^*_{n/2} + z_{1-\beta/2}\frac{\bar\sigma(\hat x^*_{n/2})}{\sqrt {n/2}}\right]
\end{equation}
where $z^*_{n/2}, \hat x^*_{n/2}, \sigma(\hat x^*_{n/2})$ are computed as before using first half of the data $\{\xi_1,\ldots,\xi_{n/2}\}$, $\bar{z}^*_{n/2}=(2/n)\sum_{i=\frac{n}{2}+1}^nH(\hat x^*_{n/2};\xi_i)$ is the evaluation of $\hat x^*_{n/2}$ using the remaining half of the data, and $\bar\sigma(\hat x^*_{n/2})=\sqrt{(1/(n/2-1))\sum_{i=\frac{n}{2}+1}^n(H(\hat x^*_{n/2}; \xi_i)-\bar z^*_{n/2})^2}$ is the empirical standard deviation at $\hat{x}^*_{n/2}$. Note that $ \bar{z}^*_{n/2}$ is a high biased estimator of $z^*$, and thus the CI (\ref{CLT2 CI}) alleviates the under coverage issue; on the other hand,  the effective sample size is reduced by half, and thus the estimates are less accurate especially when the data size is small, which may in turn affect the coverage probability of the CI. 

Due to the limited data size, we use the single replication procedure (SRP) proposed in \cite{Bayraksan2006} to estimate CIs on the optimality gap. For a given solution $\hat{x}$ that is independent of the data, the SRP outputs a one-sided ($1-\beta$) CI on the optimality gap given by
$$
\left[0, \hat{\mathcal G}_n(\hat{x}) + z_{1-\beta}\frac{\tilde\sigma(\hat x_n^*)}{\sqrt n}\right],
$$
where as before $\hat x^*_n$ is the empirical optimal solution, $\hat{\mathcal G}_n(\hat{x})= (1/n)\sum_{i=1}^n\left(H(\hat x, \xi_i)-H(\hat x^*_n,\xi_i)\right)^2$, and $$\tilde \sigma^2(\hat x_n^*) = \frac{1}{n-1}\sum_{i=1}^n\left[(H(\hat x,\xi_i)-H(\hat x_n^*,\xi_i))-(\hat h(\hat x) - \hat z^*_n) \right]^2,$$ where $\hat z^*_n$ is the empirical optimal value and $\hat h(\hat x) = (1/n)\sum_{i=1}^nH(\hat x,\xi_i)$. In all the examples considered below, we set $\beta = 0.05$. Note that all the above discussion holds for deterministically constrained problems. Nonetheless, we also apply these methods in a stochastically constrained problem as a benchmark (which is heuristic since there has been no formal proof of their validity in this case).

Note that the EL method consists of solving a max-min and a min-min problem. Supposing that the original problem \eqref{obj} or \eqref{obj constrained} is convex, then the max-min program is convex. In our examples we use the built-in Matlab solvers. The min-min program, on the other hand, is more challenging because the outer optimization involves minimizing the concave function $\min_{x\in\Theta}\sum_{i=1}^nw_iH(x;\xi_i)$ over $w$. This is not a convex problem in general. However, fixing either $w$ or $x$, optimizing over the other variable becomes a convex problem. Thus one approach is to do alternating minimization, by iteratively minimizing $w$ and $x$ while fixing each others, until no improvement is observed. Such type of schemes has appeared in chance-constrained programming (e.g., \cite{chen2010cvar,zymler2013distributionally,jiang2012data}), and it appears to work well in our examples despite a lack of global convergence guarantee. 

\subsection{Quadratic Optimization}\label{sec:quad opt}
We consider a simple unconstrained problem of minimizing a quadratic function
\begin{equation}\label{quad}
\min_{x} E[(x-\xi)^2],
\end{equation}
where $\xi$ follows an unknown distribution $F^c$. It is easy to see that the optimal solution is $x^* = E[\xi]$ and the optimal value is $z^* = Var(\xi)$.  We set $F^c$ as a standard normal distribution, and thus $x^*=0$ and $z^*=1$.

Assuming we are given $n$ observations from the normal distribution, we implement the different methods to obtain $95\%$ confidence bounds for the optimal value of \eqref{quad}. We test on three cases where we randomly generate $n= 10, 50, 100$ data points from $F^c$. For each case, we repeat the experiment $100$ times, and note down the empirical coverage probability, mean upper and lower bounds, and the mean and standard deviation of the interval width for each method. The results are summarized in Table \ref{comparisons simple}.

\begin{table}[h]
\centering
 {\small

 \begin{tabular}{cc|ccccc}
 & &Coverage & Mean lower & Mean upper & Mean interval & Standard deviation
\\
& &probability & bound & bound & width & of interval width\\
\hline
$n=10$ & EL& $0.79$ &$0.32$ &$1.89$ &$1.57$ &$0.95$\\
       & CLT&$0.72$ &$0.19$ &$1.69$ &$1.50$ &$0.93$ \\
       & CLT2&$0.80$ &$0.01$ &$2.42$ &$2.41$ &$1.75$ \\ \hline
$n=50$ & EL&$0.97$ &$0.29$ &$1.96$ &$1.67$ &$1.07$\\
       & CLT&$0.84$ &$0.60$ &$1.31$ &$0.71$ &$0.19$ \\
       &CLT2&$0.87$ &$0.45$ &$1.53$ &$1.08$ &$0.51$ \\ \hline
$n=100$& EL&$0.99$ &$0.66$ &$1.43$ &$0.77$ &$0.27$\\
       & CLT&$0.84$ &$0.70$ &$1.22$ &$0.52$ &$0.10$ \\
       &CLT2&$0.88$ &$0.60$ &$1.33$ &$0.74$ &$0.29$ \\
\end{tabular}
}
\caption{Confidence intervals on optimal values of the quadratic optimization problem}
\label{comparisons simple}
\end{table}

To compare EL and SRP on optimality gap, we first generate a solution $\hat{x}$ and evaluate its true optimality gap using a large sample size ($10^8$). Then for each of the three cases $n=10, 50, 100$, we repeat the experiment  $100$ times for each method to obtain $95\%$ confidence bounds and estimate their empirical coverage probabilities. The results are summarized in Table \ref{tab:quad_optgap}, where the suboptimal solution $\hat x = 0.62$ and its corresponding optimality gap is $0.39$.

\begin{table}[h]
\centering
 {\small

 \begin{tabular}{cc|ccccc}
 & &Coverage & Mean lower & Mean upper & Mean interval & Standard deviation
\\
& &probability & bound & bound & width & of interval width\\
\hline
$n=10$ & EL& $0.95$ &$0.06$ &$1.95$ &$1.89$ &$0.94$\\
       & CLT-SRP&$0.86$ &$0$ &$1.16$ &$1.16$ &$0.72$ \\ \hline
$n=50$ & EL&$0.99$ &$0.09$ &$1.42$ &$1.32$ &$1.54$\\
       & CLT-SRP&$0.93$ &$0$ &$0.72$ &$0.72$ &$0.26$ \\ \hline
$n=100$& EL&$0.97$ &$0.13$ &$0.83$ &$0.70$ &$0.50$\\
       & CLT-SRP&$0.92$ &$0$ &$0.57$ &$0.57$ &$0.16$ \\
\end{tabular}
}
\caption{Confidence intervals on optimality gaps of the quadratic optimization problem}
\label{tab:quad_optgap}
\end{table}


\subsection{CVaR Estimation}
In this example, we consider estimating CVaR$_{\alpha, F^c}(\xi)$, the $\alpha$-level conditional-value-at-risk of a random variable $\xi$, which we assume follows an unknown distribution $F^c$. This can be rewritten as a stochastic optimization problem:
\begin{equation}
\min_{x\in\mathbb{R}} \left\{x + \frac{1}{1-\alpha}E[(\xi-x)^+]\right\},\label{CVaR opt}
\end{equation}
where $(\cdot)^+$ is short for $\max(\cdot,0)$. We set $F^c$ as a standard normal distribution and $\alpha=0.9$. As the previous example in Section \ref{sec:quad opt}, we run the experiment $100$ times for each method and each case of $n=10, 50, 100$. The results are summarized in Table \ref{tab:CVaR} and \ref{tab:CVaR_optgap}. Note that the true optimal value can be accurately calculated and is equal to $1.755$; the suboptimal solution in this experiment is $0.71$ with optimality gap $0.36$.

\begin{table}[h]
\centering
 {\small

 \begin{tabular}{cc|ccccc}
 & &Coverage & Mean lower & Mean upper & Mean interval & Standard deviation
\\
& &probability & bound & bound & width & of interval width\\
\hline
$n=10$ & EL&$0.39$ &$0.80$&$1.65$&$0.85$&$0.56$\\
       & CLT&$0.50$&$0.95$&$2.16$&$1.21$&$1.02$ \\
       & CLT2&$0.47$&$1.14$&$4.03$&$2.89$&$4.17$ \\ \hline
$n=50$ & EL&$0.90$ &$1.21$&$2.31$&$1.10$&$0.40$\\
       & CLT&$0.81$&$1.23$&$2.16$&$0.93$&$0.35$ \\
       &CLT2&$0.78$&$0.99$&$2.55$&$1.56$&$1.22$ \\ \hline
$n=100$& EL&$0.98$ &$1.34$&$2.28$&$0.94$&$0.27$\\
       & CLT&$0.86$&$1.35$&$2.06$&$0.71$&$0.20$ \\
       &CLT2&$0.88$&$1.19$&$2.35$&$1.15$&$0.63$ \\
\end{tabular}
}
\caption{Confidence intervals on optimal values of the CVaR estimation problem}
\label{tab:CVaR}
\end{table}

\begin{table}[h]
\centering
 {\small

 \begin{tabular}{cc|ccccc}
 & &Coverage & Mean lower & Mean upper & Mean interval & Standard deviation
\\
& &probability & bound & bound & width & of interval width\\
\hline
$n=10$ & EL& $0.96$ &$0.02$ &$3.49$ &$3.47$ &$2.08$\\
       & CLT-SRP&$0.83$ &$0$ &$1.70$ &$1.70$ &$1.38$ \\ \hline
$n=50$ & EL&$1.00$ &$0.03$ &$1.46$ &$1.43$ &$0.50$\\
       & CLT-SRP&$0.85$ &$0$ &$0.81$ &$0.81$ &$0.43$\\ \hline
$n=100$& EL&$0.99$ &$0.07$ &$1.05$ &$0.98$ &$0.24$\\
       & CLT-SRP&$0.91$ &$0$ &$0.71$ &$0.71$ &$0.26$ \\
\end{tabular}
}
\caption{Confidence intervals on optimality gaps of the CVaR estimation problem}
\label{tab:CVaR_optgap}
\end{table}


\subsection{Portfolio Optimization}
Our last example considers minimizing the CVaR risk associated with the loss of an investment, subject to the condition that the expected return should exceed a certain threshold. Let's denote by $x = [x^1,\ldots,x^d]'$ the vector of holding proportions in $d$ assets, $\xi=[\xi^1,\ldots,\xi^d]'$ the random vector of asset returns, and $r_b$ the threshold for expected return.  We assume short selling is not allowed. The problem can be written as
\begin{equation} \label{CVaRmean}
\begin{array}{ll}
  \min_{x} &  CVaR_{\alpha}(-\xi'x)  \\
  \text{subject to} & E[\xi'x] \geq r_b \\
   & \sum_{i=1}^dx_i = 1 \\
   & x_i\geq 0, i=1,\ldots,d \\
\end{array}
\end{equation}
We can rewrite the problem in the form of (\ref{obj constrained}) as
\begin{equation} \label{CVaRmean}
\begin{array}{ll}
  \min_{x,c} & c+ \frac{1}{1-\alpha}E[(-\xi'x-c)^+]  \\
  \text{subject to} & E[\xi'x] \geq r_b \\
   & \sum_{i=1}^dx_i = 1 \\
   & x_i\geq 0, i=1,\ldots,d \\
\end{array}
\end{equation}

The parameter setting is as follows: $\xi$ follows a normal distribution with mean $\mu = [0.8, 1.2]'$ and covariance $\Sigma = [1~ 0; 0~ 4]$; the minimum expected return is $r_b = 1$; the CVaR level is $\alpha = 0.9$, and the confidence level is $1-\beta = 0.95$. It is easy to verify that the optimal solution to (\ref{CVaRmean}) is $x^* = [0.5, 0.5]'$, and the associate optimal value can be evaluated by Monte Carlo simulation with a large number ($10^8$) of samples, which yields $z^*\approx 0.96$. For comparison, we also implement the CLT and 2-sample CLT methods by computing the CIs according to (\ref{CLT CI}) or (\ref{CLT2 CI}); though the validity of these schemes has not been proved, we use them as heuristic to provide a benchmark. For each case of $n=10, 50, 100$, we repeat the experiment 100 times, and summarize the numerical results in Table \ref{tab:portfolio} and \ref{tab:portfolio_optgap}. In this experiment, the suboptimal solution is $[0.21 0.79]$ with optimality gap $0.73$.
\begin{table}[h]
\centering
 {\small


 \begin{tabular}{cc|ccccc}
 & &Coverage & Mean lower & Mean upper & Mean interval & Standard deviation
\\
& &probability & bound & bound & width & of interval width\\
\hline
$n=10$ & EL & $0.26$ & $-0.01$ & $1.04$ & $1.05$ & $4.66$\\
       & CLT& $0.21$ & $0.33$ & $1.21$ & $0.88$ & $1.61$ \\
       &CLT2& $0.52$ & $0.32$ & $5.46$ & $5.14$ & $6.90$ \\ \hline
$n=50$ & EL &$0.69$  & $0.14$ & $1.20$ & $1.06$ & $0.49$ \\
       & CLT& $0.58$ & $0.60$ & $1.85$ & $1.24$ & $0.64$ \\
       &CLT2& $0.67$ & $0.34$ & $2.45$ & $2.12$ & $1.59$\\ \hline
$n=100$& EL & $0.74$ & $0.28$ & $1.38$ & $1.09$ & $0.60$\\
       &CLT & $0.59$ & $0.73$ & $1.63$ & $0.90$ & $0.34$ \\
       &CLT2& $0.64$ & $0.63$ & $2.02$ & $1.39$ & $0.80$ \\
\end{tabular}
}
\caption{Confidence intervals on optimal values of the portfolio optimization problem}
\label{tab:portfolio}
\end{table}

\begin{table}[h]
\centering
 {\small

 \begin{tabular}{cc|ccccc}
 & &Coverage & Mean lower & Mean upper & Mean interval & Standard deviation
\\
& &probability & bound & bound & width & of interval width\\
\hline
$n=10$ & EL& $0.56$ &$0.57$ &$2.44$ &$1.87$ &$1.58$\\
       & CLT-SRP& $0.79$ &$0$ &$1.92$ &$1.92$ &$1.56$ \\ \hline
$n=50$ & EL&$0.90$ &$0.33$ &$1.91$ &$1.59$ &$1.60$\\
       & CLT-SRP&$0.57$ &$0$ &$0.92$ &$0.92$ &$0.82$\\ \hline
$n=100$& EL&$0.92$ &$0.48$ &$1.50$ &$1.02$ &$0.50$\\
       & CLT-SRP&$0.59$ &$0$ &$0.87$ &$0.87$ &$0.61$ \\
\end{tabular}
}
\caption{Confidence intervals on optimality gaps of the portfolio optimization problem}
\label{tab:portfolio_optgap}
\end{table}

\subsection{Summary of Numerical Results}

We note in all three examples EL in general has the highest coverage probability on optimal values except when the data size is very small ($n=10$). Although EL has wider intervals than the direct CLT method, its interval widths are often comparable to or smaller than the 2-sample CLT method, which usually has higher coverage probability than the plain CLT method. EL also has higher coverage probabilities on the optimal gap than SRP, accompanied by wider intervals than SRP. Overall speaking, EL performs competitively compared to the CLT methods.

One thing worth mentioning is that the empirical converge probability in the last example is in general smaller compared to the previous two examples. A potential reason (in addition to the assumptions or the validity of the compared methods not being rigorously justified) is the higher dimensionality that naturally requires more data to achieve a similar level of accuracy in the SAA solution. 
Nevertheless, we can see that the EL method still produces CIs with a higher coverage probability than CLT methods when the data size is not too small, and the coverage probability improves as the data size increases.


\section{Conclusion}
We have studied the EL method to construct statistically valid CIs for the optimal value and the optimality gap of a given solution for stochastic optimization problems. The method builds on positing two optimization problems that resemble DRO problems with Burg-entropy divergence ball constraints, with the ball size suitably calibrated by a $\chi^2$-quantile with a suitable degree of freedom. We have studied the theory leading to the statistical guarantees and numerically compared our method to approaches suggested by the CLT. Built on a rigorous foundation, our method provides a competitive method for evaluating the statistical uncertainty for stochastic optimization problems under limited data. In future work, we plan to further refine the accuracy of our method.

\section*{Appendix}
\begin{lemma}[Lemma 11.2 in \cite{owen2001empirical}]
Let $Y_i$ be i.i.d. random variables in $\mathbb R$ with $EY_i^2<\infty$. We have $\max_{1\leq i\leq n}|Y_i|=o(n^{1/2})$ a.s..
\end{lemma}


\bibliographystyle{abbrv}
\bibliography{bibliography2}

\end{document}